\documentclass{eptcs}
\providecommand{\event}{GandALF 2021}
\usepackage{amsfonts}
\usepackage{amsmath}
\usepackage{amsthm}
\usepackage{amssymb}
\usepackage[utf8]{inputenc}
\usepackage[english]{babel}

\usepackage{scalerel,stackengine}
\stackMath
\newcommand\reallywidehat[1]{%
\savestack{\tmpbox}{\stretchto{%
  \scaleto{%
    \scalerel*[\widthof{\ensuremath{#1}}]{\kern-.6pt\bigwedge\kern-.6pt}%
    {\rule[-\textheight/2]{1ex}{\textheight}}
  }{\textheight}%
}{0.5ex}}%
\stackon[1pt]{#1}{\tmpbox}%
}
\theoremstyle{definition}
\newtheorem{definition}{Definition}[section]

\theoremstyle{plain}
\newtheorem{theorem}{Theorem}[section]
\newtheorem{lemma}{Lemma}[section]
\newtheorem{proposition}{Proposition}[section]
\newtheorem{corollary}{Corollary}[section]

\title{Finite Model Property and Bisimulation for LFD}
\author{Raoul Koudijs
\institute{ILLC \\ Amsterdam, The Netherlands}
\email{raoulxluoar@gmail.com}}
\def\titlerunning{Finite Model Property and Bisimulation for Local Logics of Dependence}
\def\authorrunning{Raoul Koudijs}
\begin{document}
\maketitle

\begin{abstract}
Recently, Baltag \& van Benthem introduced a decidable logic of functional dependence (LFD) that extends the logic of Cylindrical Relativized Set Algebras (CRS) with atomic local dependence statements. Its semantics can be given in terms of generalised assignment models or their modal counterparts, hence the logic is both a first-order and a modal logic. We show that LFD has the finite model property (FMP) using Herwig's theorem on extending partial isomorphisms, and prove a bisimulation invariance theorem characterizing LFD as a fragment of first-order logic.
\end{abstract}

\section{\textbf{Introduction}}
Recently, Baltag \& van Benthem introduced a decidable logic of functional dependence (LFD) that extends the logic of Cylindrical Relativized Set Algebras (CRS) \cite{ModLangBoundFragmPredLog} with atomic dependence statements. The semantics is given in terms of dependence models\footnote{These are just the generalised assignment models known from \cite{ModLangBoundFragmPredLog}\cite{Exploring_Lo_Dyn}.}, which are pairs $(M,A)$ of a first-order structure $M$ together with a \textit{fixed} set of variable assignments (or 'team') $A\subseteq M^V$ on $M$, where $V$ is some (possibly finite) ambient set of variables. Formulas are evaluated at \textit{individual} assignments $s\in A$; in particular the dependence atoms get the following semantics: $s\models D_Xy$ if for all $t\in A$, $s\restriction X=t\restriction X$ implies $s(y)=t(y)$. This is in contrast with logics based on team semantics, where dependence formulas are evaluated at teams and this team is dynamically changed over the course of evaluation. Whereas most logics based on team semantics are undecidable, non-classical and have expressive going beyond FOL, LFD is decidable with a classical semantics and can be considered a fragment of FOL.

Many interesting notions of dependence (such as lineair dependence in vector spaces, temporal dependence in dynamical systems and strategic interaction in a multi-player game) can be formalized in LFD \cite{BaltagvBenthemLFD}. Moreover, LFD invites a natural epistemic interpretation where (sets of) variables may represent (groups of) agents, (joint) questions or objects.\cite{BaltagvBenthemLFD} The dependence modalities then capture distributed knowledge or the interrogative modality, while the dependence atoms capture epistemic superiority or inquisitive implication (or other 'mixed' notions). More spectacularly, \cite{Learning_What_Others_Know} introduces a complete and decidable dynamic-epistemic logic based on LFD with so-called 'reading events' as well as a notion of 'common distributed knowledge' that combines features of common knowledge and distributed knowledge.

Dependence models are closely related to relational databases: assignments are rows in the table and each variable represents a column, or attribute. Here is a simple numerical example of a dependence model, viewed as a database:
\begin{center}
\begin{tabular}{|c|c|c|} 
 \hline
 x & y & z\\ [0.5ex]
 \hline
 1 & 0 & 1\\
 1 & 0 & 0\\
 0 & 1 & 1\\
 2 & 0 & 2\\
 \hline
\end{tabular}
\end{center}
In this table we see that e.g. $y$ locally depends on $x$ in the first row, because the second row, which agrees on $x$ with the first, also agrees on $y$ with it (and no other rows agree on $x$ with the first row). In fact, this dependence holds at all rows, in which case we say that $y$ \textit{globally} depends on $x$. Conversely, $x$ does \textit{not} depend globally on $y$ because it does not locally depend on $y$ at the first row: both $1$ and $2$ occur as $x$-values of rows that share the current $y$-value $0$. Finally, because the fourth row is the only row with $z$-value $2$, all other variables locally depend on $z$ there.

The foregoing example witnesses the close connection between LFD and the study of dependence in databases, and indeed the Projection and Transitivity axioms of LFD recapture Armstrong's Axioms for functional dependence \cite{BaltagvBenthemLFD}. Deeper connections with database theory as well as team semantics might arise by introducing dynamics on the level of teams, generalizing the semantics to dependence universes, i.e. families of dependence models \cite{BaltagvBenthemLFD}. In particular, dependence models 

The decidability proof in \cite{BaltagvBenthemLFD} uses completeness of LFD w.r.t a purely syntactic 'type semantics' resembling the 'quasi-models' studied in connection with the Guarded Fragment \cite{Exploring_Lo_Dyn}\cite{ModLangBoundFragmPredLog}. The question whether LFD has the finite model property (FMP)
w.r.t. dependence models remained an open problem \cite{BaltagvBenthemLFD}. Our main result is that LFD has the FMP, by a new application of Herwig's theorem on extending partial isomorphisms. Moreover, we define dependence bisimulations and show that LFD can be characterized as the fragment of FOL that is invariant under this notion. Independently, another notion of bisimulation for LFD along more standard lines has been proposed in \cite{local_deps}. We show that these notions are equivalent, but that dependence bisimulations suggest a more efficient procedure for checking bisimilarity.
\section{Preliminaries}
We first introduce the language LFD, dependence models and type models. A pair $(V,\tau)$, where $V$ is set of variables and $\tau$ is a relational language is called a \textit{vocabulary}. When both $V$ and $\tau$ are finite, we say that $(V,\tau)$ is a finite vocabulary. We write $FOL[V,\tau]$ for the set of first-order formulas with variables in $V$ (both free and bound) and predicates in $\tau$, and similarly for $LFD[V,\tau]$. We assume that each vocabulary becomes equipped with an arity map $ar:\tau\to\mathbb{N}$.
\begin{definition}{(\textbf{Syntax})} The language $LFD[V,\tau]$ is recursively defined by:
\[\varphi::= P\mathbf{x}\;|\;\neg\varphi\;|\;\varphi\wedge\varphi\;|\;\mathbb{D}_X\varphi\;|\;D_Xy\]
where $X\subseteq V$ is a \textit{finite} set of variables, $y\in V$ an individual variable, $P\in\tau$ a predicate symbol and $\mathbf{x}=(x_1,...,x_n)\in V^{ar(P)}$ a finite string of variables.\footnote{LFD as a modal language is generated by the same definition, but where $D_Xy(\;),P\mathbf{x}(\;)$ become unary predicates in $\tau$.} Fixing notation, for any $Y\subseteq V$, we write $s\models D_XY$ if $s\models D_Xy$ holds for all $y\in Y$. We also skip the set brackets for singletons, writing $D_xY$ for $D_{\{x\}}Y$ , and $D_xy$ for $D_{\{x\}}\{y\}$. For every $\varphi\in LFD$, we define its free variables by:
\begin{itemize}
    \item $Free(Px_1...x_n)=\{x_1,...,x_n\}$
    \item $Free(D_Xy)=Free(\mathbb{D}_X\varphi)=X$
    \item $Free(\neg\varphi)=Free(\varphi)$,\; $Free(\varphi\wedge\psi)=Free(\varphi)\cup Free(\psi)$
\end{itemize}
\end{definition}
Moreover, we let $V_{\varphi}$ denote the set of variables \textit{occurring} in $\varphi$. This is in general a superset of the free of variables, i.e. $V_{D_Xy}=X\cup\{y\}$. Further, we say that $\tau_{\varphi}:=\{P\in\tau\;|\;P\;\textrm{occurs in}\;\varphi\}$.
\begin{definition}{(\textbf{Dependence Models})} A \textit{dependence model} (for the vocabulary $(V,\tau)$) $\mathbb{M}$ is a pair $\mathbb{M}=(M,A)$ of a relational structure $M$ for $\tau$, together with a fixed team $A\subseteq O^V$.\footnote{We use letters $M$ for first-order structures and blackboard bold letters $\mathbb{M}=(M,A)$ for dependence models.} For each $X\subseteq V$, we define an agreement relation $=_X$ on the team:
\[s=_Xt\;\textrm{iff}\;s\restriction X=t\restriction X\]
Note that $V$ may be finite. We call a dependence model \textit{distinguished} if all the assignments are injective.
\end{definition}

\begin{definition}{(\textbf{Semantics})} Truth of a formula $\varphi$ in a dependence model $\mathbb{M}=(M,A)$ at an assignment $s\in A$ is defined by the following clauses  (the Boolean cases are defined as usual:
\begin{align*}
& s\models P\mathbf{x}\;\textrm{iff}\;s(\mathbf{x})\in I^{\mathbb{M}}(P)\\
& s\models\mathbb{D}_X\varphi\;\textrm{iff}\;t\models\varphi\;\textrm{holds for all}\;t\in A \;\textrm{with}\; s =_X t\\
& s\models D_Xy\;\textrm{iff}\;s=_X t\;\textrm{implies}\; s=_yt\;\textrm{for all}\;t\in A.
\end{align*}
\end{definition}
Where $s(\mathbf{x})$ denotes the tuple $(s(x_1),...,s(x_m))$ for $\mathbf{x}=(x_1,...,x_m)$. Clearly, for every dependence model $(M,A)$ and assignments $s,t\in A$, there is a unique set $V^{s,t}:=\{v\in V\;|\;s=_vt\}$ that is \textit{the maximal set of variables on which $s,t$ agree}. An important feature of the semantics is LFD satisfies Locality.
\[\textrm{\textbf{Locality}}:\;\;\textrm{If}\;\; s=_Xt\;\textrm{and}\;Free(\varphi)\subseteq X,\;\;\textrm{then}\;\; s\models\varphi\;\textrm{iff}\;t\models\varphi\]
Next to dependence models, LFD is weakly complete w.r.t. a non-standard type semantics.\cite{BaltagvBenthemLFD} In other words, only LFD over finite vocabularies is complete for this semantics. Type models were used in \cite{BaltagvBenthemLFD} as a technical auxiliary to prove completeness and decidability. Types are defined relative to closures. We obtain the \textit{closure} $\Psi:=Cl(\psi)$ of a formula $\psi$ by adding to $\{\psi\}$ all formulas $D_Xy$ for $X\cup\{y\}\subseteq V_{\psi}$ and closing the resulting set under subformulas and single negation. \footnote{For every non-negated formula $\varphi$ (i.e. a formula whose principal connective is not $\neg$) we add $\neg\varphi$ to the closure, and for negated formulas we we do nothing. The resulting closure set will not contain any formulas with double negations.} 

\begin{definition}{($\Psi$-\textbf{Types})} Let $\Psi$ be a closure in $LFD[V,\tau]$. A subset $\Sigma\subseteq\Psi$ is a $\Psi$-\textit{type} if it satisfies the following conditions (where all formulas mentioned run over $\Psi$ only):
\begin{itemize}
    \item [(a)] $\neg\psi\in\Sigma$ iff $\psi\not\in\Sigma$
    \item [(b)] $(\psi\wedge\chi)\in\Sigma$ iff $\psi\in\Sigma$ and $\chi\in\Sigma$
    \item [(c)] if $\mathbb{D}_X\psi\in\Sigma$, then $\psi\in\Sigma$
    \item [(d)] $D_Xx\in\Sigma$ for all $x\in X\subseteq V$
    \item [(e)] $D_XY,D_YZ\in\Sigma$ implies $D_XZ\in\Sigma$
\end{itemize}
\end{definition}
For $X\subseteq V_{\psi}$, we define a relation $\sim_X$ on types $\Sigma,\Delta\subseteq\Psi$:
\begin{align*}
\Sigma\sim_X\Delta\qquad\textrm{iff}\qquad&\{\phi\in\Sigma\;|\;Free(\phi)\subseteq D^{\Sigma}_X\}=\{\phi\in\Delta\;|\;Free(\phi)\subseteq D^{\Sigma}_X\}
\end{align*}
where $D^{\Sigma}_X=\{y\in V_{\varphi}\;|\;D_Xy\in\Sigma\}$ is the \textit{dependence-closure} of $X$ w.r.t and $\Sigma$. Observe that $\Sigma\sim_X\Delta$ implies $D^{\Sigma}_X=D^{\Delta}_X$ as $Free(D_Xy)=X$. 

\begin{definition}{(\textbf{Type Models})} A \textit{type model} (for $\Psi$) is a family of $\Psi$-types satisfying:
\begin{itemize}
    \item if $\neg\mathbb{D}_X\neg\psi\in\Sigma\in\mathfrak{M}$, then there exists a $\Delta\in\mathfrak{M}$, such that $\psi\in\Delta$ and $\Sigma\sim_X\Delta$.
    \item $\Sigma\sim_{\emptyset}\Delta$ holds for all $\Sigma,\Delta\in\mathfrak{M}$.
\end{itemize}
Type models are always finite, as there are only finitely many $\Psi$-types for a given closure $\Psi$. This proves the decidability as LFD is (weakly) complete w.r.t. type models \cite{BaltagvBenthemLFD}. The semantic conditions for type models are given by membership:
\[\Delta\models\psi\qquad\textrm{iff}\qquad\psi\in\Delta\]
\end{definition}
\subsection{\textbf{Tree Model Property}}
Every satisfiable LFD formula can be satisfied on a certain tree-like dependence model. This fact follows from the fact that dependence models and type models provide equivalent semantics for LFD, i.e. each type model can be represented as a dependence model and vice versa \cite{BaltagvBenthemLFD}. The interesting direction is representing arbitrary type models as dependence models by means of an \textit{unravelling} construction in the sense of modal logic.
To say what we mean by 'tree-like' we need the graph-theoretic notion of a $k$-tree (the definition is taken from \cite{grAdel_1999}). Say that an $r$-tuple of objects $\mathbf{a}$ from a $\tau$-structure $M$ is \textit{live} in $M$, if there is some $r$-ary $P\in\tau$ such that $M\models P\mathbf{a}$.

\begin{definition}{(\textbf{$k$-Tree})} A $\tau$-structure $M$ is a \textit{$k$-tree} if there exists a tree (i.e. an acyclic, connected graph) $T=(V,E)$ and a function $F:V\to\{M'\subseteq M\;|\;|M'|\leq k\}$, assigning to every node $v\in V$ of $T$ a set $F(v)$ of at most $k$ elements of $M$, such that the following two conditions hold.
\begin{itemize}
    \item [(i)] For every live tuple $\mathbf{a}=(a_1,...,a_r)$ from $M$, there is some node $v$ such that $\{a_1,...,a_r\}\subseteq F(v)$.
    \item [(ii)] For every element $a$ of $M$, the set of nodes $\{v\in V\;|\;a\in F(v)\}$ is connected (and hence induces a subtree of $T$).
\end{itemize}
$M$ is of finite branching degree if $T$ is, that is if the set of neighbours of every node in $T$ is finite.
\end{definition}

\begin{theorem}{\textbf{Representation of Type Models} \cite{BaltagvBenthemLFD}}\\
Let $\mathfrak{M}$ be a type model for $\Psi$. There exists a dependence model $\mathbb{M}=(M,A)$ with $\mathfrak{M}=\{type_{\Psi}(s)\;|\;s\in A\}$.
\end{theorem}
\begin{proof}
Let $m=|\mathfrak{M}|$ and $k=|V|$ where $V$ is the set of variables occurring in formulas in $\Psi$ (i.e. $V=\bigcup\{V_{\psi}\;|\;\psi\in\Psi\}$). Fix a type $\Sigma_0\in\mathfrak{M}$. A \textit{good path} is a sequence $\pi=\langle\Sigma_0,X_1,...,X_n,\Sigma_n\rangle$ with $n>0$ such that for each $i\leq n$ (a) $\Sigma_i\in\mathfrak{M}, X_i\subseteq V$ and (b) $\Sigma_{i-1}\sim_{X_i}\Sigma_i$. Write $last(\pi)=\Sigma_n$ for the last element of $\pi$, and $lh(\pi)=n+1$ for the length of $\pi$ (not counting the variable sets). For each good path $\pi$, we define the \textit{path assignment} $v_{\pi}$, assigning objects of the form $(\pi,v)$ to variables $v\in V$:
\begin{align}
    & v_{\pi}(v)=(\pi,v)\;\textrm{if}\;\pi\;\textrm{has length 1, i.e.}\;\pi=\langle\Sigma_0\rangle\;\textrm{is the root of our tree}.\\
    & v_{\pi}(v)=v_{\pi'}(v)\;\textrm{if}\;\pi=(\pi',X,\Sigma)\;\textrm{with}\;v\in D^{last(\pi')}_X\\
    & v_{\pi}(v)=(\pi,v)\;\textrm{if}\;\pi=(\pi',X,\Sigma)\;\textrm{with} v\not\in D^{last(\pi')}_X
\end{align}
So new objects are created whenever the value for a variable is not locally determined by the predecessor path. We obtain a team $A:=\{v_{\pi}\;|\;\pi\;\textrm{a good path}\}$ on the structure $M$ with domain $\bigcup_{v_{\pi}\in A}v_{\pi}[V]$ and where an $r$-ary $P\in\tau$ holds of an $r$-tuple $((\pi_1,x_1),...,(\pi_r,x_r))$ iff all paths $\pi_i$ are linearly ordered by initial segment and the formula $P\mathbf{x}\in last(\pi_j)$, where $\pi_j$ is the longest path amongst $\{\pi_1,...,\pi_r\}$.

This yields a \textit{distinguished} dependence model $\mathbb{M}=(M,A)$ whose objects $(\pi,x)$ are typed by a unique variable $x$. The set of all good paths, ordered by initial segment, forms a tree $T$ whose branching degree is bounded by $2^k\times m$. Together with the map $\pi\mapsto v_{\pi}[V]$, this shows that $M$ is a $k$-tree of finite branching degree. Finally, we have the following crucial \textit{truth lemma} \cite{BaltagvBenthemLFD}:
\begin{lemma}{\textbf{Truth Lemma}}\\
For all formulas $\psi\in\Psi$ and good paths $\pi:\;\;$ $\mathbb{M},v_{\pi}\models\psi$ iff $\psi\in last(\pi)$
\end{lemma}
This lemma implies that $\mathfrak{M}=\{type_{\Psi}(s)\;|\;s\in A\}$: because every type $\Delta\in\mathfrak{M}$ occurs as the $last(\pi)$ for some unique good path of length 2 already, namely $\pi_{\Delta}:=(\Sigma_0,\emptyset,\Delta)$. Moreover, note that we are free to choose the initial fixed type $\Sigma_0$ from $\mathfrak{M}$ in the definition of good path, and hence, by the truth lemma, we can choose what type to be satisfied at the root.
\end{proof}

\begin{corollary}{\textbf{Tree Model Property}}\\
If $\psi\in LFD$ is satisfiable and $|V_{\psi}|=k$, there is a dependence model $\mathbb{M}=(M,A)$, where $M$ is $k$-tree of finite branching degree, satisfying $\varphi$ at the root assignment.
\end{corollary}
\begin{definition}{(\textbf{First-Order Translation})}
Although interpreted over a generalised semantics, LFD in \textit{finitely} many variables can be encoded back into FOL over standard structures. So let $V$ be a finite set of variables with enumeration $\mathbf{v}=(v_1,...,v_n)$. We double the amount of variables, creating a set of copied variables $V'$ from the variables in $V$. We ensure that the relevant assignments agree on their values for variables $v$ and their copies $v'$ by the conjunction $\mathbf{v}=\mathbf{v'}$.\footnote{This additional condition (it was not in the original formulation in \cite{BaltagvBenthemLFD}) is essential for encoding the semantics of the dependence atoms into FOL, which treats as variables as completely independent otherwise.} Further, we introduce a new $n$-ary predicate $A$ such that $A\mathbf{v}$ encodes the fact that the tuples of values assigned to $\mathbf{v}$ by the current assignment is the range of some admissible assignment from the team (this is a tuple because $V$ is finite). The first-order translation $tr:LFD[V,\tau]\to FOL[V\cup V',\tau]$ is defined by \cite{BaltagvBenthemLFD}:
\begin{itemize}
    \item $tr(P\mathbf{x}) = P\mathbf{x}$ and $tr$ commutes with Boolean connectives
    \item $tr(\mathbb{D}_X\psi) = \forall\mathbf{z}(A\mathbf{v}\to tr(\psi))$, where $\mathbf{v}$ is the enumeration of all the variables in $V$ and $\mathbf{z}$ is the enumeration of all the variables in $V-X$.
    \item $tr(D_Xy):= \forall\mathbf{z}\forall\mathbf{z'}((A\mathbf{v}\wedge A\mathbf{v}[\mathbf{z'}/\mathbf{z}])\to y = y')$, where $\mathbf{v},\mathbf{z}$ are as in part (d), $\mathbf{z'}$ and $y'$ are the corresponding fresh $V'$-copies of $\mathbf{z}$ and $y$ respectively.\footnote{Furthermore, $A\mathbf{v}[\mathbf{z'}/\mathbf{z}]$ denotes the formula that is obtained by replacing the variables $\mathbf{z}$ by $\mathbf{z'}$ in the formula $A\mathbf{v}$}.
\end{itemize}
There is a one-to-one correspondence between dependence models and structures in this extended language. If $\mathbb{M}=(M,A)$ is a dependence model, $T(\mathbb{M})$ is the expansion of $M$ with the interpretation $I(A):=\{s(\mathbf{v})\;|\;s\in A\}$. Conversely, given any $\tau\cup\{A\}$-structure $M'$ we obtain a team $A:\{s:V\to M'\;|\;s(\mathbf{v})\in I^{M'}(A)\}$ which together with a reduct of $M'$ makes for the corresponding dependence model. We have the equivalence:
\[\mathbb{M},s\models\varphi\qquad\textrm{iff}\qquad T(\mathbb{M}),s^+ \models\mathbf{v}=\mathbf{v'}\to tr(\varphi)\]
for every $s\in A$ and all assignments $s^+\in M^{\mathrm{Var}}$ extending $s$. This translation easily adapts to other local dependence atoms proposed in \cite{local_deps}, e.g. $tr(x=y):=\;x=y$ and $tr(\mathbf{x}\in\mathbf{y}):=\; \exists\mathbf{v'}(A\mathbf{v'}\wedge\bigwedge_{i\leq|\mathbf{x}|}x_i=y'_i)$.
\end{definition}

\section{\textbf{Characterization}}
The original paper \cite{BaltagvBenthemLFD} left finding a bisimulation-invariance theorem characterizing LFD as an open problem. precisely which formulas in $FOL[V\cup V',\tau\cup\{A\}]$ are equivalent to the $tr$-translation of an LFD-formula over standard structures.\footnote{There is also a \textit{modal translation} of LFD into FOL that extends the well-known standard translation of modal logic into the 2-variable fragment of FOL. A similar characterization theorem can be proved via this translation and the relational semantics for LFD, as our notion of bisimulation as well as the one proposed in \cite{local_deps} are naturally formulated on dependence models as well as their modal counterparts.} The following notion of \textit{dependence bisimulation} exactly characterizes $LFD$ as the largest fragment of $FOL$ invariant under this notion. Say that a set of variables $X$ is \textit{dependence-closed} at $s'$ if $D^{s'}_X:=\{y\in V\;|\;;s'\models D_Xy\}=X$, or equivalently if $s'\models D_Xy$ implies $y\in X$.
\begin{definition}{(\textbf{Dependence Bisimulation})} Let $\mathbb{M},\mathbb{M'}$ be dependence models. We say that a non-empty relation $Z\subseteq A\times A'$ is a \textit{dependence-bisimulation} if for every $(s,s')\in Z$:
\begin{itemize}
    \item [(\textbf{Atom})] $s\models P\mathbf{x}$ iff $s'\models P\mathbf{x}$
    \item [(\textbf{Forth})] For every $t\in A$, (i) the set $V^{s,t}$ is dependence-closed at $s'$ and\newline 
    there is some $t'\in A$ such that (ii) $s'=_{V^{s,t}}t'$ and (iii) $(t,t')\in Z$
    \item [(\textbf{Back})] symmetric to the (Forth) clause
\end{itemize}
\end{definition}
Dependence bisimulations are always \textit{total}; every state is related to another by the bisimulation.
\begin{proposition} LFD-formulas are invariant under dependence bisimulations.
\end{proposition}
\begin{proof}
Let $\mathbb{M},\mathbb{M'}$ be dependence models and $Z\subseteq A\times A'$ a dependence bisimulation with $(s,s')\in Z$ and $\varphi\in$ LFD. We show that $s\models\varphi$ iff $s'\models\varphi$ by induction on the complexity of $\varphi$; the atomic and Boolean cases are trivial. For the other cases, we show only one direction.

($\mathbb{D}_X\psi)\quad$ Suppose that $s\models\mathbb{D}_X\psi$ and let $s'=_Xt'$, i.e. $X\subseteq V^{s',t'}$, for some $t'\in A'$. By the (Back)-clause there is some $t\in A$ such that $s=_Xt$ and $(t,t')\in Z$. Hence $t\models\psi$ and so $t'\models\psi$ by $(IH)$.

($D_Xy)\quad$ Suppose that $s\models D_Xy$ and let $s'=_Xt'$ for some $t'\in A'$. We want to show that $s'=_yt'$, i.e. $y\in V^{s',t'}$. By the (Back)-clause there is some $t\in A$ with $(t,t')\in Z$, $s=_{V^{s',t'}}t$ and $V^{s',t'}$ is dependence-closed at $s$. As $X\subseteq V^{s',t'}$, by monotonicity of dependence we have $s\models D_{V^{s',t'}}y$. This shows that $y\in V^{s',t'}$ as $V^{'s,t'}$ is dependence-closed at $s$.
\end{proof}

Dependence bisimulations in fact characterize LFD as a fragment of FOL. This can be shown by formulating an analogue of dependence bisimulations for structures of the form $T(\mathbb{M})$, and showing that on $\omega$-saturated structures of this form, LFD-equivalence implies dependence-bisimilarity.

Independently, another notion of bisimulation characterizing LFD has been proposed in \cite{local_deps} that treats dependence atoms like ordinary relational atoms. That is, instead of the dependence-closed condition they simply require that "$s\models D_Xy$ iff $s'\models D_Xy$" holds for all $X\cup\{y\}\subseteq V$. It follows that proposition 3.1 shows that dependence bisimulations are also bisimulations in their sense. Conversely, "$s\models D_Xy$ iff $s'\models D_Xy$" clearly implies the dependence-closed condition, hence the two notions are equivalent. It follows that the proof given in \cite{local_deps} also shows that LFD is the dependence bisimulation-invariant fragment of FOL.
\begin{theorem}{\textbf{Van Benthem Characterization}}\\
$LFD$ is the largest fragment of $FOL$ that is invariant under dependence bisimulations.
\end{theorem}

Dependence bisimulations suggest a more efficient way to implement a \textit{bisimilarity-checking} algorithm for LFD compared to the definition in \cite{local_deps}. For what proposition 3.1 shows is that, given $(M,A),(M',A')$ with $s\in A,s'\in A'$, it actually suffices to check that "$s\models D_Xy$ iff $s'\models D_Xy$ for all $y\in V$" holds for all $X\in\{V^{s,t}\subseteq V\;|\;t\in A\}\cup\{V^{s',t'}\subseteq V\;|\;t'\in A'\}$ in order to conclude that "$s\models D_Xy$ iff $s'\models D_Xy$ for all $y\in V$" holds for all $X\subseteq V$. This could be used to avoid an exponential blow-up in $|V|$.

Dependence bisimulations generalise naturally to extensions of LFD. For instance, we can extend $LFD$ with the equality relation $=$, yielding the logic $LFD^=$ which was shown to be a conservative reduction class of FOL and hence undecidable in \cite{local_deps}. Dependence bisimulations with an extended (Atom) clause that also ranges over equality can be shown to characterize $LFD^=$ as a fragment to FOL. Interestingly, over \textit{full} dependence models (i.e. those $(M,A)$ with $A=M^V$, which are standard first-order structures repackaged as dependence models), dependence bisimulations (for LFD over a finite vocabulary $(V,\tau)$ with $|V|=k$) coincides with $k$-potential isomorphism, which characterizes first-order logic in $k$ variables.

\section{\textbf{Finite Model Property}}
We show that LFD has the FMP w.r.t the intended dependence model semantics, by an application of Herwig's theorem similar to the one in \cite{grAdel_1999}. Fix a satisfiable LFD-formula $\varphi$, and let $\Phi:=Cl(\{\varphi\})$. We let $(V,\tau):=(V_{\varphi},\tau_{\varphi})$ be the smallest vocabulary containing $\varphi$ and hence $\Phi$. Note that $(V,\tau)$ is a finite vocabulary, so let $|V|=k$. We know that there is a tree-like dependence model $\mathbb{M}=(M,A)$, with associated tree $T$ of good paths, satisfying $\varphi$ at the root assignment. Furthermore, the degree of $T$ is bounded by $m\times 2^k$, where $m$ is the number of distinct $\Phi$-types. Our strategy is as follows: we will cut the underlying $k$-tree $M$ to a finite structure, encode the dependence atoms in a richer language and finally use Herwig's theorem to generate out of this a finite dependence model that is bisimilar to the original tree-model. Define a sub-team of $A$ by:
\[A_{cut}:=\{v_{\pi}\in A\;|\;lh(\pi)\leq 3\}\]
and let $M_{cut}$ be the submodel of $M$ induced by $\bigcup\{v_{\pi}[V]\subseteq M\;|\;v_{\pi}\in A_{cut}\}$; we call $\mathbb{M}_{cut}:=(M_{cut},A_{cut})$ the \textit{cut-off} model. This is a finite model because the branching degree of $T$ is bounded and $V$ is finite. The truth lemma clearly no longer holds on this cut-off model, because some existential witnesses are missing for assignments of length 3.

We extend the language to include an $|X|$-ary relation $R^{X,y}$ for each $X\cup\{y\}\subseteq V$, and obtain the (still finite) richer language $\tau^+\supseteq\tau$. We will use these relations to encode the semantics of the dependence atoms. We expand the structure $M_{cut}$ underlying the cut-off model to a $\tau^+$ structure by putting:

\[I^{M_{cut}}(R^{X,y}):=\{v_{\pi}(\mathbf{x})\;|\;D_Xy\in last(\pi)\}\]

so that $\mathbb{M}_{cut},v_{\pi}\models R^{X,y}\mathbf{x}$ iff $D_Xy\in last(\pi)$. In the end, we want to show that $R^{x,y}\mathbf{x}\leftrightarrow D_Xy$ holds on the Herwig extension, so that we can recover an appropriate dependence model from it. To show this, we will need the following restricted version of this claim on the cut-off model:

\begin{proposition}
For each $v_{\pi}\in A_{cut}$ of length $lh(\pi)\leq 2:\quad v_{\pi}\models D_Xy\to R^{X,y}\mathbf{x}$.
\end{proposition}
\begin{proof}
By contraposition, so suppose that $v_{\pi}\not\models R^{X,y}\mathbf{x}$. This means that $D_Xy\not\in last(\pi)$, so for the good path $\pi^+:=(\pi,X,last(\pi))$ (it is a good path as $last(\pi)\sim_Xlast(\pi)$ trivially holds) we have that $v_{\pi}=_Xv_{\pi^+}$ and $v_{\pi}\ne_yv_{\pi^+}$, i.e. $v_{\pi}\not\models D_Xy$.
\end{proof}
Herwig's theorem on extending partial isomorphism \cite{Herwig1998ExtendingPI} is a result about first-order relational languages. It tells us that any finite structure with some set of partial isomorphisms on it has a finite extension in which all these partial isomorphisms extend to automorphisms. This theorem has already been used to show the FMP of the Guarded Fragment (GF) \cite{grAdel_1999}.
\begin{theorem}{\textbf{Herwig}}\\
Let $\sigma$ be a finite relational language, $C$ a finite $\sigma$-structure and $\{p_1,...,p_k\}$ a (finite) set of partial isomorphisms on $C$. Then there exists a \textit{finite} extension $C^+$ of $C$ that satisfies the following conditions:
\begin{itemize}
    \item [(i)] Every $p_i$ extends to a unique automorphism $\widehat{p_i}$ of $C^+$. This yields a subgroup $\langle\widehat{p_1},...,\widehat{p_k}\rangle$ of the automorphism group of $C^+$.
    \item [(ii)] If a tuple $\mathbf{a}=(a_1,....,a_r)$ from $C^+$ is live or $r=1$, then there exists an automorphism $f\in\langle\widehat{p_1},...,\widehat{p_k}\rangle$ such that for each $i\leq r$, $f(a_i)\in C$.
    \item [(iii)] If $\exists f\in\langle\widehat{p_1},...,\widehat{p_k}\rangle$ and $a,b\in C$ such that $f(a)=b$, then either $f=id$ or there is a unique $p\in\langle p_1,...,p_k\rangle$ such that $\widehat{p}=f$ and $p(a)=b$.
\end{itemize}
\end{theorem}
where $\langle p_1,...,p_k\rangle$ is the collection of all partial isomorphisms that can be obtained by composing the $p_i$ with their inverses. Note that $\langle p_1,...,p_k\rangle$ is strictly speaking \textit{not} a group as it need not be the case that $p\circ p^{-1}$ is the identity on $C$ (in general, it is the identity on a subset of $C$).

Condition (iii) is in need of further clarification. In words, it says that elements in the submodel $C$ are only mapped to each other by some $f\in\langle f_1,...,f_n\rangle$ if this is forced given the choice of partial isomorphisms. Uniqueness of $p$ in this condition is ensured by the fact that the map $\widehat{(\;)}$ extends to a bijective map $\widehat{(\;)}:\langle p_1,...,p_k\rangle\to\langle\widehat{p_1},...,\widehat{p_k}\rangle$ that commutes with the operations $\circ,(\;)^{-1}$ (and the identity $id$). By condition (i), $\widehat{(\;)}$ is defined on the subset $\{p_1,...,p_k\}$. Set $\widehat{p^{-1}}:=\widehat{p}^{-1}$ and $\widehat{p\circ p'}:=\widehat{p}\circ\widehat{p'}$; so commutation follows by definition. It immediately follows that the map is injective. For surjectivity, let $f\in\langle\widehat{p_1},...,\widehat{p_k}\rangle$. By definition $f=\widehat{p_{i_1}}^{\epsilon_1}\circ...\circ \widehat{p_{i_m}}^{\epsilon_m}$ for some $\{i_1,...,i_m\}\subseteq\{1,...,k\}$ and $\epsilon_j\in\{-1,1\}$ for each $j\leq m$. Define $p:=p_{i_1}^{\epsilon_1}\circ...\circ p_{i_m}^{\epsilon_m}\in\langle p_1,...,p_k\rangle$. Now observe:
\[\widehat{p}=\reallywidehat{p_{i_1}^{\epsilon_1}\circ...\circ p_{i_m}^{\epsilon_m}}=\widehat{p_{i_1}^{\epsilon_1}}\circ...\circ\widehat{p_{i_m}^{\epsilon_m}}=\widehat{p_{i_1}}^{\epsilon_1}\circ...\circ\widehat{p_{i_m}}^{\epsilon_m}=f\]
We proceed with specifying a choice of partial isomorphisms on the cut-off model. If $\pi$ is a good path of $lh(\pi)=3$ and $last(\pi)=\Delta$, then there is a partial isomorphism $p_{\pi}:v_{\pi}[V_{\varphi}]\to v_{\pi_{\Delta}}[V_{\varphi}]$ such that $p_{\pi}\circ v_{\pi}=v_{\pi_{\Delta}}$, where $\pi_{\Delta}:=\langle\Sigma_0,\emptyset,\Delta\rangle$ so $lh(\pi_{\Delta})=2$. We pick the finite set of partial isomorphisms $\{p_{\pi}\;|\;\pi\;\textrm{good path of}\;lh(\pi)=3\}=\{p_1,...,p_k\}$. The following proposition tells us what kind of partial isomorphisms 
are in $\langle p_1,...,p_k\rangle$.

\begin{lemma} If $p\in\langle p_1,...,p_k\rangle$ with $pv_{\pi}=_Xv_{\pi'}$, then there are $v_{\rho},v_{\rho'}\in A_{cut}$ with $last(\rho)=last(\rho')$ such that $v_{\rho}=_Xv_{\pi}$, $v_{\rho'}=_Xv_{\pi'}$ and $pv_{\rho}=v_{\rho'}$.
\end{lemma}
\begin{proof}
Let $p\in\langle p_1,...,p_k\rangle$ such that $pv_{\pi}=_Xv_{\pi'}$. By definition, $p=p_{i_m}^{\epsilon_m}\circ...\circ p_{i_1}^{\epsilon_1}$ for some $\{i_1,...,i_m\}\subseteq\{1,...,k\}$ and $\epsilon_j\in\{-1,1\}$ for each $1\leq j\leq m$. Note that for each $j\leq m$ we have that $p_{i_j}\in\{p_1,...,p_k\}=\{p_{\pi}\;|\;\pi\;\textrm{a good path of}\;lh(\pi)=3\}$, so $p_{i_j}^{\epsilon_j}\circ v_{\pi_{j-1}}=v_{\pi_j}$ \footnote{More specifically  $p_{i_j}^{\epsilon_j}\circ v_{\pi_{j-1}}=_Vv_{\pi_j}$, but this is the same as equality as $dom(v_{\pi_j})=dom(v_{\pi_{j_1}})=V$. Another way of putting this is that $dom(p_{i_j}^{\epsilon_j})=v_{\pi_{j-1}}[V]$ and $cod(p_{i_j}^{\epsilon_j})=v_{\pi_j}[V]$.} for some $v_{\pi_{j-1}},v_{\pi_j}\in A_{cut}$ with $last(\pi_{j-1})=last(\pi_j)$. In particular, there are $v_{\pi_0},v_{\pi_1}\in A_{cut}$ such that $last(\pi_0)=last(\pi_1)$ and $p_{i_1}^{\epsilon_1}v_{\pi_0}=v_{\pi_1}$. Set $\rho:=\pi_0$. It follows that $v_{\pi}=_Xv_{\pi_0}$ and so $pv_{\pi_0}=_Xpv_{\pi}=_Xv_{\pi'}$, i.e.
\[pv_{\pi_0}=p_{i_m}^{\epsilon_m}\circ...\circ p_{i_1}^{\epsilon_1}v_{\pi_0}=_Xv_{\pi'}\]
This was the base case for an inductive argument up to $m$. So let $j\leq m$ and suppose that $v_{\pi_j}\in A_{cut}$ with $last(\pi_j)=last(\pi_0)$ and 
\[pv_{\pi_0}=p_{i_1}^{\epsilon_1}\circ ... \circ p_{i_{j+1}}^{\epsilon_{j+1}}v_{\pi_j}=_Xv_{\pi'}\]
Now recall that $p_{i_{j+1}}^{\epsilon_{j+1}}v_{\pi_j}=v_{\pi_{j+1}}$ for some $v_{\pi_{j+1}}\in A_{cut}$ with $last(\pi_{j+1})=last(\pi_j)$. Moreover, it follows that $p_{i_1}^{\epsilon_1}\circ...\circ p_{i_{j+2}}^{\epsilon_{j+2}}v_{\pi_{j+1}}=_Xv_{\pi'}$. Hence by induction, there is some $v_{\pi_m}\in A_{cut}$ with $pv_{\pi_0}=v_{\pi_m}$ such that $last(\pi_m)=last(\pi_0)$ and
\[v_{\pi_m}=p_{i_m}^{\epsilon_m}\circ...\circ p_{i_1}^{\epsilon_1}v_{\pi_0}=pv_{\pi_0}=_Xv_{\pi'}\]
then for $\rho=\pi_0$ and $\rho'=\pi_m$ we have proved the lemma 
\end{proof}
The associated first-order structure $T(\mathbb{M}_{cut})$ of the Herwig extension is a finite model in a finite relational language $\tau^+\cup\{A\}$, and $\{p_1,...,p_k\}$ is a finite set of partial isomorphisms on it. Hence, by Herwig's theorem, there exists a \textit{finite} extension $T(\mathbb{M}_{cut})^+$ of this structure, the \textit{Herwig extension}, satisfying conditions (i)-(iii) w.r.t $\{p_1,...,p_k\}$. It is easy to see that the Herwig extension corresponds in the canonical way (i.e. see the first-order translation above) to a dependence model $\mathbb{M}_{cut}^+:=(M_{cut}^+,A_{cut}^+)$ such that $T(\mathbb{M}_{cut}^+)=T(\mathbb{M}_{cut})^+$. Recall that we want to establish a bisimulation between the finite Herwig extension $\mathbb{M}_{cut}^+$ and the infinite tree model $\mathbb{M}$. To do this, we will need the following lemmas.

\begin{lemma}{\textbf{Level 2 Lemma}}\\ For every $s\in A_{cut}^+$ there is an $f\in\langle\widehat{p_1},...,\widehat{p_k}\rangle$ such that $f\circ s=v_{\pi}\in A_{cut}$ where $lh(\pi)\leq 2$.
\begin{proof}
Let $s\in A_{cut}^+$. Then the tuple $s(\mathbf{v})\in I(A)$ is live in $T(\mathbb{M}_{cut}^+)$. Hence by condition (ii) there is some automorphism $f\in\langle\widehat{p_1},...,\widehat{p_k}\rangle$ such that $fs(\mathbf{v})$ is a tuple of objects of the submodel $T(M_{cut})$. As $f$ is an isomorphism, it follows that $fs(\mathbf{v})\in I(A)$ as well. But this can only be if $fs(\mathbf{v})=v_{\pi}(\mathbf{v})$ for some $v_{\pi}\in A_{cut}$. Now suppose that $lh(\pi)=3$, with $last(\pi)=\Delta$, then by (i) there is an automorphism $\widehat{p_{\pi}}$ such that $\widehat{p_{\pi}}f\in\langle\widehat{p_1},...,\widehat{p_k}\rangle$ and $\widehat{p_{\pi}}fs=v_{\pi_{\Delta}}$, where $lh(\pi_{\Delta})=2$. Hence we may assume that there exists some $g\in\langle\widehat{p_1},...,\widehat{p_k}\rangle$ such that $gs=v_{\pi}$ for some path assignment $v_{\pi}\in A_{cut}$ of length $lh(\pi)\leq 2$.
\end{proof}
\end{lemma}
Next, we generalise the notion of 'underlying type' (i.e. $last(\pi)$ for a path assignment $v_{\pi}$) to all assignments in $A_{cut}^+$. We define a function $type(\;):A_{cut}^+\to\{\Delta\subseteq\Phi\;|\;\Delta\;\textrm{is a}\;\Phi\textrm{-type}\}$. Set $type(v_{\pi}):=last(\pi)$ for all $v_{\pi}\in A_{cut}\subset A_{cut}^+$. For $s\in A_{cut}^+\setminus A_{cut}$, by the level 2 lemma we know there is $f\in\langle\widehat{p_1},...,\widehat{p_k}\rangle$ such that $fs=v_{\pi}\in A_{cut}$, and we set $type(s):=last(\pi)$.

\begin{proof}{\textbf{Well-definedness of} $type(\;)$}\\ 
Let $f,g\in\langle\widehat{p_1},...,\widehat{p_k}\rangle$ be automorphisms with $fs=V_{\pi}\in A_{cut}$ and $gs=v_{\pi'}\in A_{cut}$. Observe that $f\circ g^{-1}$ is an automorphism in the subgroup $\langle\widehat{p_1},...,\widehat{p_k}\rangle$ that maps elements in $M_{cut}$ to each other, as $fg^{-1}\circ v_{\pi'}=v_{\pi}$. Hence by (iii) there must be a unique $p\in\langle p_1,...,p_k\rangle$ such that $\widehat{p}=fg^{-1}$ and thus $pv_{\pi'}=v_{\pi}$. Lemma 4.1 tells us that there are assignments $v_{\rho},v_{\rho'}\in A_{cut}$ with $v_{\pi}=_V v_{\rho}$, $v_{\pi'}=_V v_{\rho'}$ and $last(\rho)=last(\rho')$. It is an easy consequence of the Truth Lemma (lemma 2.1) and Locality that $v_{\pi}=_Vv_{\rho}$ implies that $last(\pi)=last(\rho)$ and similarly for $\pi',\rho'$.\footnote{For suppose that $v_{\pi}=_Vv_{\rho}$. Observe that $D^{v_{\pi_0}}_V=V$ for any path assignment with $dom(v_{\pi_0})=V$. By Locality the hypothesis gives that $\{\xi\;|\;v_{\pi}\models\xi\;\&\;Free(\xi)\subseteq V\}=\{\xi\;|\;v_{\rho}\models\xi\;\&\;Free(\xi)\subseteq V\}$. By the Truth Lemma, this in turn implies that $\{\xi\;|\;\xi\in last(\pi)\;\&\;Free(\xi)\subseteq V\}=\{\xi\;|\;\xi\in last(\rho)\;\&\;Free(\xi)\subseteq V\}$ which says that $last(\pi)\sim_V last(\rho)$, but this clearly implies that $last(\pi)=last(\rho)$.} Hence $last(\pi)=last(\rho)=last(\rho')=last(\pi')$.
\end{proof}
This last fact used, i.e. that $v_{\pi}=_Xv_{\pi'}$ implies $last(\pi)\sim_X last(\pi')$, we will now generalise to all assignments in $s,t\in A_{cut}^+$ w.r.t their 'underlying types' $type(s),type(t)$.
\begin{lemma}{\textbf{Type Lemma}}\\
If $s,t\in A_{cut}^+$ with $s=_Xt$, then $type(s)\sim_X type(t)$.
\end{lemma}
\begin{proof}
Let $s,t\in A_{cut}^+$ with $s=_Xt$. By the level 2 lemma, there is $f\in\langle\widehat{p_1},...,\widehat{p_k}\rangle$ such that $fs=v_{\pi}\in A_{cut}$ with $lh(\pi)\leq 2$, so $type(s)=last(\pi)$. As $f$ is an isomorphism on $T(\mathbb{M}_{cut}^+)$, we know that $ft\in A_{cut}^+$ is an assignment as well, with $fs=v_{\pi}=_Xft$. By applying the level 2 lemma again to $ft$, we get a $g\in\langle\widehat{p_1},...,\widehat{p_k}\rangle$ such that $gft=v_{\pi'}\in A_{cut}$ with $lh(\pi')\leq 2$, so $type(t)=last(\pi')$. Again, we know that $gv_{\pi}\in A_{cut}^+$ is also an assignment (though in general not one in $A_{cut}$) such that $gv_{\pi}=gfs=_Xgft=v_{\pi'}$. But observe that the automorphism $g$ maps $v_{\pi}(x)\mapsto v_{\pi'}(x)$ for all $x\in X$, hence by condition (iii) there must be a unique $p\in\langle p_1,...,p_k\rangle$ such that $\widehat{p}=g$ and so $pv_{\pi}=_Xv_{\pi'}$. By Lemma 4.1, there are $v_{\rho},v_{\rho'}\in A_{cut}$ such that $v_{\pi}=_Xv_{\rho}$, $v_{\pi'}=_Xv_{\rho'}$ and $last(\rho)=last(\rho')$. Invoking the Truth Lemma and Locality as before this implies that $last(\pi)\sim_Xlast(\rho)$ and $last(\pi')\sim_X last(\rho')$. Concatenating these facts we see that \[type(s)=last(\pi)\sim_X last(\rho)=last(\rho')\sim_X last(\pi')=type(t)\]
\end{proof}
\begin{lemma}{\textbf{Encoding Lemma}}\\
For all $s\in A_{cut}^+$ and all $R^{X,y}\in\tau^+:\quad s\models R^{X,y}\mathbf{x}\leftrightarrow D_Xy$
\end{lemma}
\begin{proof}
($\leftarrow$) By the level 2 lemma, there is $f\in\langle\widehat{p_1},...,\widehat{p_k}\rangle$ such that $fs=v_{\pi}\in A_{cut}$ with $lh(\pi)\leq 2$. Applying the first-order translation to proposition 3.1 we get that $T(\mathbb{M}_{cut}),v_{\pi}\models tr(\neg R^{X,y}\mathbf{x}\to\neg D_Xy)$. But observe that 
\begin{align*}
    tr(\neg R^{X,y}\mathbf{x}\to\neg D_Xy)\;=\;\neg R^{X,y}\mathbf{x}\to tr(\neg D_Xy)\;\equiv\;& R^{X,y}\mathbf{x}\vee\exists\mathbf{z},\mathbf{z'}(A\mathbf{v}\wedge A\mathbf{v'}[\mathbf{z}/\mathbf{z'}]\wedge y\ne y')\\
    \;\equiv\;&\exists\mathbf{z},\mathbf{z'}(R^{X,y}\mathbf{x}\vee(A\mathbf{v}\wedge A\mathbf{v'}[\mathbf{z}/\mathbf{z'}]\wedge y\ne y'))
\end{align*}
is an \textit{existential} first-order formula. Hence by the dualized version of the Łoś-Tarski theorem, this still holds in the Herwig \textit{extension}, i.e. $T(\mathbb{M}_{cut}^+),v_{\pi}\models\neg R^{X,y}\mathbf{x}\to tr(\neg D_Xy)$. As $f$ is an isomorphism on $T(\mathbb{M}_{cut}^+)$ and $fs=v_{\pi}$, we get that $T(\mathbb{M}_{cut}^+),s\models\neg R^{X,y}\mathbf{x}\to tr(\neg D_Xy)$, as desired.

($\to$) Suppose that $s\models R^{X,y}\mathbf{x}$, and let $s=_Xt$ for some $t\in A_{cut}^+$. The former fact implies that $D_Xy\in type(s)$ and the latter by the Type Lemma implies that $type(s)\sim_Xtype(t)$. It follows that $D_Xy\in type(t)$ as well. Applying the level 2 lemma two times successively as before, we obtain automorphism $f,g\in\langle\widehat{p_1},...,\widehat{p_k}\rangle$ such that $fs=v_{\pi}\in A_{cut}^+$, $gft=v_{\pi'}\in A_{cut}^+$ with $V^{s,t}=V^{v_{\pi},ft}=V^{gv_{\pi},v_{\pi'}}$ (recall the notation $V^{a,b}=\{v\in V\;|\;a=_vb\}$). As in the type lemma, we see that $g:v_{\pi}(x)\mapsto v_{\pi'}(x)$ (i.e. $gv_{\pi}=_Xv_{\pi'}$) for all $x\in X$ and thus by condition (iii) there is a unique $p\in\langle p_1,...,p_k\rangle$ such that $\widehat{p}=g$ and hence $pv_{\pi}=_Xv_{\pi'}$.

We know by Lemma 4.1 that there must be $v_{\rho},v_{\rho'}\in A_{cut}$ with $last(\rho)=last(\rho')$ such that $v_{\pi}=_Xv_{\rho}$, $v_{\pi'}=_Xv_{\rho'}$ and $pv_{\rho}=v_{\rho'}$. By fact 4.9 from \cite{BaltagvBenthemLFD}, this implies that there is a path $last(\pi)\sim_X....\sim_Xlast(\rho)$ and similarly for $\pi',\rho'$. We saw that $D_Xy\in type(s)\cap type(t)=last(\pi)\cap last(\pi')$, so in fact $D_Xy$ must be in all the types along these paths. But then it follows from condition (2) of the recursive definition of path assignments (in the proof of theorem 2.1) that $v_{\pi}=_yv_{\rho}$ and $v_{\pi'}=_yv_{\rho'}$. But recall that $pv_{\rho}=v_{\rho'}$ so:
\[ft=g^{-1}gft=g^{-1}v_{\pi'}=\widehat{p^{-1}}v_{\pi'}=_y\widehat{p^{-1}}v_{\rho'}=v_{\rho}\]
But then $v_{\pi}=_yv_{\rho}=_yft$ so by transitivity $y\in V^{v_{\pi},ft}=V^{s,t}$ and we conclude that $s=_yt$.
\end{proof}

\begin{theorem} The dependence models $\mathbb{M}$ and $\mathbb{M}_{cut}^+$ are dependence-bisimilar.
\end{theorem}
\begin{proof}
We show that the relation $Z\subseteq A_{cut}^+\times A$ defined by $Z:=\{(s,v_{\pi})\;|\;type(s)=last(\pi)\}$ is an LFD-bisimulation in the sense of $\cite{local_deps}$ and hence, by our remark above, also a dependence bisimulation. Pick an arbitrary pair $(s,v_{\pi})\in Z$. By the level 2 lemma, there is some $f\in\langle\widehat{p_1},...,\widehat{p_k}\rangle$ such that $fs=v_{\pi'}\in A_{cut}$ with $lh(\pi')\leq 2$, hence $type(s)=v_{\pi'}$. As $type(\;)$ is well-defined, it follows that $last(\pi)=last(\pi')$. We show that the pair $(s,v_{\pi})$ satisfies (Atom) (i.e. the one which also ranges over dependence atoms \cite{local_deps}) and is closed under the (Back) \& (Forth) clauses (without the dependence-closedness condition).

$\textbf{(Atom)}$ Observe that the chain of equivalences:

\[s\models_{\mathbb{M}_{cut}^+} P\mathbf{x}\quad \textrm{iff}\quad v_{\pi'}\models_{\mathbb{M}_{cut}^+} P\mathbf{x}\quad \textrm{iff}\quad P\mathbf{x}\in last(\pi')=last(\pi)\quad \textrm{iff}\quad v_{\pi}\models_{\mathbb{M}} P\mathbf{x}\]

holds for every $P\in\tau^+$ (i.e. including the relations $R^{X,y}$!) by the fact that $f$ is an isomorphism with $fs=v_{\pi'}$ and the way we have specified the interpretation $I(P)$ on both models. Invoking the encoding lemma, this implies that $\mathbb{M}_{cut}^+,s\models D_Xy$ iff $\mathbb{M},v_{\pi}\models D_Xy$.

\textbf{(Forth)} Let $t\in A_{cut}^+$ be some assignment in the Herwig extension, and let $V^{s,t}$ be the maximal set of variables on which $s$ and $t$ agree. By the Type Lemma $type(s)\sim_{V^{s,t}} type(t)$. But $last(\pi)=last(\pi')=type(s)$, so it follows that $\pi^+:=(\pi,V^{s,t},type(t))$ is a good path. Clearly $v_{\pi^+}\in A$ with $v_{\pi}=_{V^{s,t}} v_{\pi^+}$, and lastly $(t,v_{\pi^+})\in Z$ as $type(t)=last(\pi^+)$.

\textbf{(Back)} Let $v_{\pi''}\in A$, with $V^{\pi,\pi''}=\{v\in V\;|\;v_{\pi}=_vv_{\pi''}\}$ the maximal
set ot variables on which $v_{\pi},v_{\pi''}$ agree. By a now familiar argument involving the Truth lemma and Locality (i.e. the analogue of the Type Lemma for $\mathbb{M}$), it follows that $last(\pi)\sim_{V^{s,t}}last(\pi'')$. As $type(s)=last(\pi')=last(\pi)$, we see that $\pi'_+:=(\pi',V^{\pi,\pi''},last(\pi''))$ is a good path. Moreover, we know that $lh(\pi')\leq 2$ which implies that $lh(\pi'_+)=lh(\pi')+1\leq 2+1=3$ and so $v_{\pi'_+}\in A_{cut}$ is in the cut-off model. Clearly $v_{\pi'}=_{V^{\pi,\pi''}}v_{\pi'_+}$. Set $t:=f^{-1}v_{\pi'_+}$, then $s=_{V^{s,t}}t$ as $fs=v_{\pi'}$ and moreover $(t,v_{\pi''})\in Z$ since $type(t)=last(\pi'^+)=last(\pi'')$.
\end{proof}

\begin{corollary}{\textbf{Bounded Model Property}}\\
Every satisfiable $\varphi$ in LFD has a finite model whose size is bounded by a computable function of $\varphi$. \footnote{Any formula $\varphi$ determines a unique smallest \textit{finite} vocabulary $(V,\tau)$ such that $Cl(\varphi)$ belongs to $LFD[V,\tau]$; the computable function takes as input $|V|$, the maximal arity $r$ of relations in $\tau$, and the number of distinct $\Phi$-types $m$.}
\end{corollary}
\begin{proof}
Let $\varphi$ be a satisfiable LFD-formula with closure $\Phi$ in the language $(V,\tau)$ and $|V|=k$. By the tree model property, there is a $k$-tree $M$ and a team $A$ such that $\mathbb{M}=(M,A)$ is a dependence model satisfying $\varphi$ at the root assignment. We cut this tree at length 3 and obtain the cut-off model whose size is upper bounded by $k(b+b^2+b^3)$, where $b\in\mathbb{N}$ is the branching degree of the $k$-tree $\mathbb{M}$. Note that $b$ itself has $m\times 2^k$ as upper bound, where $m:=|\{\Delta\subseteq\Phi\;|\;\Delta\;\textrm{is a}\;\Phi\textrm{-type}\}|$ and $\Phi=Cl(\varphi)$. It follows that the size of the cut-off model is already exponential in the size of the variables $|V|$.

Now construct the Herwig extension $\mathbb{M}_{cut}^+=(M_{cut}^+,A_{cut}^+)$ as above. Using the bound given in \cite{Herwig1998ExtendingPI}, we get that $|M_{cut}^+|\leq itexp(2r-1,p(|M_{cut}|)$ is upper bounded by an iterated exponential of a polynomial function $p$ of degree $r$ of $|M_{cut}|$, where $r$ is the maximal arity of predicates in $\tau$. By theorem 4.2, $\mathbb{M}$ and $\mathbb{M}_{cut}^+$ are-bisimilar. As dependence bisimulations are always total, there is some assignment $s\in A_{cut}^+$ with $(s,v_{\langle\Sigma_0\rangle})\in Z$. By the invariance result above (proposition 3.1), it follows that $\mathbb{M}_{cut}^+,s\models\varphi$.
\end{proof}

\section{Conclusion}
We have introduced dependence bisimulations and have shown that this notion characterizes LFD as a fragment of FOL. Furthermore, we have shown that LFD has the finite (or bounded) model property, by a new application of Herwig's theorem and a tree-model property established in \cite{BaltagvBenthemLFD}. The same strategy can be used to carry out a direct proof of the FMP through the equivalent \textit{modal} semantics.\footnote{The proof of this can be found in an extended version of this paper (arXiv:2107.06042).} With minor adaptations, the proof goes through, though we need to appeal to a more general version of Herwig's theorem (theorem 5 in \cite{Herwig1998ExtendingPI}) to ensure that the Herwig extension is a tree in order to obtain a standard relational model from it.
By reducing the maximal arity $r$ to $2$, going through the modal semantics significantly lowers the upper bound on the size of the Herwig extension to being singly exponential in the size of the cut-off model.

While LFD only adds local dependence atoms $D_Xy$ to CRS, extensions of CRS with other local versions of atomic dependency properties have been considered in \cite{local_deps}.\footnote{We will consider only the logics defined in \cite{local_deps} that are closed under negation, i.e. those $L[\Omega]$ for which $\Omega$ is closed under negation.} The authors show that LFD extended with either equality or inclusion is undecidable and that the extension of CRS with both inclusion and equality is contained in GF. CRS with independence atoms was shown undecidable in \cite{BaltagvBenthemLFD}, resulting in a complete characterization of the satisfiability problems of such logics. The same paper also studies the \textit{model-checking} problem for such logics, and shows it to be PTIME-complete in restriction to finitely many variables. However,  this tight bound is only obtained on the assumption that the local atoms considered (i.e. inclusion, dependence, independence and equality) are all efficiently checkable.

One open problem is to determine the computational complexity of the satisfiability problem for LFD. It seems that, with a few adaptations, the satisfiability test for GF given in \cite{grAdel_1999} can be used for the case of LFD. Indeed, the 'witnesses for satisfiability' defined there closely resemble type models. A more conceptual challenge is connecting the \textit{qualitative} notion of dependence studied by LFD to probabilistic, i.e. \textit{quantitative} notions of correlation and dependence.
\newpage

\bibliography{biblio.bib}
\bibliographystyle{eptcs}
\end{document}